\documentclass[11pt,a4paper,twoside]{article}

\usepackage[latin1]{inputenc}
\usepackage[english]{babel}
\usepackage{amsmath,amssymb, enumerate,array,amsthm,graphicx,color}
\newcommand{\MC}[0]{\ensuremath{\textsc{MC}}}
\newcommand{\CSP}[0]{\ensuremath{\textsc{CSP}}}
\newcommand{\QCSP}[0]{\ensuremath{\textsc{QCSP}}}
\newcommand{\NP}[0]{\ensuremath{\mathsf{NP}}}
\newcommand{\coNP}[0]{\ensuremath{\mathsf{co\mbox{-}NP}}}

\newcommand{\Logspace}[0]{\ensuremath{\mathsf{Logspace}}}

\newcommand{\Ptime}[0]{\ensuremath{\mathsf{P}}}
\newcommand{\coPspace}[0]{\ensuremath{\mathsf{co\mbox{-}Pspace}}}
\newcommand{\Pspace}[0]{\ensuremath{\mathsf{Pspace}}}

\newcommand{\symclos}[0]{\mbox{sym-clos}}
\newcommand{\tranclos}[0]{\mbox{tran-clos}}
\newcommand{\doub}[0]{\mbox{doub}}

\newcommand{\FO}[0]{\ensuremath{\mathsf{FO}}}
\newcommand{\tuple}[1]{\ensuremath{\mathbf{#1}}}

\newcommand{\notmodels}{\ \makebox[0.1cm][l]{\ensuremath{\models}}/ \ }

\theoremstyle{plain}
\newtheorem{theorem}{Theorem}%[section]

\newtheorem{lemma}[theorem]{Lemma}
\newtheorem{proposition}[theorem]{Proposition}

\theoremstyle{definition}

\theoremstyle{remark}

\newtheorem*{examples}{Examples}

\begin{document}

\title{ \Large \bf Model Checking Positive Equality-free FO: \\ Boolean Structures and Digraphs of Size Three}
%\markboth {\hspace{3em} First-Order Model Checking Problems \hfill}{\hfill Barnaby Martin \hspace{3em}}
\author{ Barnaby Martin \\ {\small Department of Computer Science, University of Durham,} \\ {\small Science Labs, South Road, Durham DH1 3LE, U.K.} \\ {\tt b.d.martin@durham.ac.uk} }
\date{}

\maketitle

\begin{abstract}
We study the model checking problem, for fixed structures $A$, over positive equality-free first-order logic -- a natural generalisation of the non-uniform quantified constraint satisfaction problem $\QCSP(A)$. We prove a complete complexity classification for this problem when $A$ ranges over 1.) boolean structures and 2.) digraphs of size (less than or equal to) three. The former class displays dichotomy between \Logspace\ and \Pspace-complete, while the latter class displays tetrachotomy between \Logspace, \NP-complete, \coNP-complete and \Pspace-complete.
\end{abstract}

\section{Introduction}
The \emph{model checking problem} over a logic $\mathcal{L}$ -- here always a fragment of first-order logic (\FO) -- takes as input a structure (model) $A$ and a sentence $\varphi$ of $\mathcal{L}$, and asks whether $A \models \varphi$. When $\mathcal{L}$ is the \emph{existential conjunctive positive} fragment of \FO, $\{ \exists, \wedge \}$-\FO, the model checking problem is equivalent to the much-studied \emph{constraint satisfaction problem} (\CSP). Similarly, when $\mathcal{L}$ is the \emph{(quantified) conjunctive positive} fragment of \FO, $\{ \exists, \forall, \wedge \}$-\FO, the model checking problem is equivalent to the well-studied \emph{quantified constraint satisfaction problem} (\QCSP).
In this manner, the \QCSP\ is the generalisation of the \CSP\ in which universal quantification is restored to the mix. In both cases it is essentially irrelevant whether or not equality is permitted in the sentences, as it may be propagated out by substitution. Much work has been done on the parameterisation of these problems by the structure $A$ -- that is, where $A$ is fixed and only the sentence is input. It is conjectured \cite{FederVardi} that the ensuing problems $\CSP(A)$ attain only the complexities \Ptime\ and \NP-complete. This may appear surprising given that 1.) so many natural \NP\ problems may be expressed as \CSP s (see, e.g., myriad examples in \cite{jeavons98algebraic}) and 2.) \NP\ itself does not have this `dichotomy' property (assuming $\Ptime \neq \NP$) \cite{Ladner}. While this \emph{dichotomy conjecture} remains open, it has been proved for certain classes of $A$ (e.g., for structures of size at most three \cite{BulatovJACM} and for undirected graphs \cite{HellNesetril}) 
The like parameterisation of the \QCSP\ is also well-studied, and while no overarching polychotomy has been conjectured, only the complexities \Ptime, \NP-complete and \Pspace-complete are known to be attainable (for trichotomy results on certain classes see \cite{OxfordQuantifiedConstraints,DBLP:conf/cie/MartinM06}, as well as the dichotomy for boolean structures, e.g., in \cite{Nadia}).

In previous work, \cite{CiE2008}, we have studied the model checking problem, parameterised by the structure, for various fragments of \FO. Various complexity classifications are obtained and the case is put that the only interesting fragment, other than those that give rise to the \CSP\ and the \QCSP\ is \emph{positive equality-free} \FO, $\{ \exists, \forall, \wedge, \vee \}$-\FO\ (the classification for the remaining fragments in near-trivial). This model checking problem may be seen as the generalisation of the \QCSP\ in which disjunction is returned to the mix -- although note that the absence of equality is here important. 

In \cite{CiE2008}, some general hardness results are given for the model checking problem, parameterised by the structure $A$, over positive equality-free \FO, which we denote $\{ \exists, \forall, \wedge, \vee \}$-$\MC(A)$. In the case where $A$ ranges over boolean digraphs, a full classification -- a dichotomy -- is given.
In this paper, we extend this result in two directions. Firstly, in Section~\ref{sec:booleanstart}, we prove, for boolean structures $B$, that $\{ \exists, \forall, \wedge, \vee \}$-$\MC(B)$ is either in \Logspace\ or is \Pspace-complete. A similar result, but with different classification criteria is known for $\QCSP(B)$, i.e. $\{ \exists, \forall, \wedge \}$-$\MC(B)$ (see, e.g., \cite{Nadia}). Secondly, in Section~\ref{sec:digraphstart}, we prove, for digraphs $H$ of size (less than or equal to) three, that $\{ \exists, \forall, \wedge, \vee \}$-$\MC(B)$ is either in \Logspace, is \NP-complete, is \coNP-complete or is \Pspace-complete. While the classification criterion for the boolean case is fairly simple, the criteria for digraphs of size three are far from obvious (our result is achieved through a series of ad hoc methods). This suggests that, from the viewpoint of complexity theory, the class of model checking problems over $\{ \exists, \forall, \wedge, \vee \}$-\FO\ is pleasingly rich.

\section{Preliminaries}

Let $A$ be a $\sigma$-structure, for some relational signature $\sigma:=\langle R_1,\ldots,R_m\rangle$, over universe $|A|$ of cardinality $||A||$. Let $\{ \exists, \forall, \wedge, \vee \}$-\FO\ be the positive equality-free fragment of first-order logic (\FO). Define the problem $\{ \exists, \forall, \wedge, \vee \}$-$\MC(A)$ to have as input a sentence $\varphi$ of $\{ \exists, \forall, \wedge, \vee \}$-\FO, and to have as yes-instances those sentences such that $A \models \varphi$.
While $\{ \exists, \forall, \wedge, \vee \}$-\FO\ is the principal fragment involved in this paper, we will sometimes have recourse to the equality-free fragments $\{ \neg, \exists, \forall, \wedge, \vee \}$-\FO, $\{ \exists, \forall, \wedge \}$-\FO\ and $\{ \exists, \wedge, \vee \}$-\FO,\footnote{We imagine the definitions of these fragments to be clear from their notation: for example, $\{ \exists, \wedge, \vee \}$-\FO\ is the fragment of \FO\ involving no instances of negation, universal quantification or equality.} together with their respective model checking problems.

We assume that all sentences $\varphi$ of $\{ \exists, \forall, \wedge, \vee \}$-\FO\ are in prenex form, since they may be thus translated in logarithmic space. We note that $\{ \exists, \forall, \wedge, \vee \}$-$\MC(A)$, which contains $\QCSP(A)$, is always in \Pspace, by an inward evaluation procedure of the quantified variables (see \cite{VardiComplexity}). Similarly, $\{ \exists, \wedge, \vee \}$-$\MC(A)$, which contains $\CSP(A)$, is always in \NP. Henceforth all proofs of \Pspace-completeness for $\{ \exists, \forall, \wedge, \vee \}$-$\MC(A)$ will include only proof of hardness. The reductions used will involve only straightforward substitutions, and will always be logspace many-to-one.

\section{Boolean Structures}
\label{sec:booleanstart}

Let $B$ be a boolean structure, that is $||B||=2$, where we consider $|B|$ normalised as $\{0,1\}$.  
Relations of $B$ that contain no tuples (respectively, all tuples) are of little interest from the viewpoint of complexity theory, since they may be substituted in instances $\varphi$ of $\{ \exists, \forall, \wedge, \vee \}$-$\MC(B)$ by the boolean false (respectively, true) without affecting whether $B \models \varphi$. Such substitutions may be carried out in logarithmic space, and we note here that, if each relation of $B$ is either empty or contains all tuples, then $\{ \exists, \forall, \wedge, \vee \}$-$\MC(B)$ is in \Logspace, since evaluation is equivalent to the Boolean Sentence Value Problem \cite{NLynch} (note that this would be the case were $||B||$ to be $1$). Henceforth, in this section, we work under the assumption 
\begin{itemize}
\item[$( \dag )$] that $B$ does not contain any relations that are either empty or contain all tuples.
\end{itemize}

Define the \emph{canonical relation} $R_B$ (note the subscript) to be $R^B_1 \times \ldots \times R^B_m$ (where $R^B_i$ is the interpretation of $R_i$ in $B$). The arity of $R_B$ is the sum of the arities of $R^B_1,\ldots, R^B_m$. $R_B$ effectively encodes all the relations of $B$; note that the stipulation that $B$ contains no empty relations is essential to its definition.

In a boolean structure $B$, whose canonical relation of arity $r$ is $R_B(v_1,\ldots,v_r)$, $0$ is termed a \emph{$\forall$-canon} and $1$ a \emph{$\exists$-canon} if, for all partitions $I | J$ of $\{v_1,\ldots,v_r\}$,
\[ B \ \models \ {R_B}(I/0,J) \rightarrow {R_B}(I/1,J), \]
where ${R_B}(I/0,J)$ and ${R_B}(I/1,J)$ are $R_B$ with the variables of $I$ substituted by $0$ and $1$, respectively (note that ${R_B}(I/0,J)$ and ${R_B}(I/1,J)$ contain $|J|$ free variables, and the statement should hold for all their instantiations). In the parlance of, e.g., \cite{LaroseLotenTardif}, $0$ and $1$ being $\forall$-canon and $\exists$-canon, respectively, is equivalent to $1$ \emph{dominating} $0$.
We may term $1$ a $\forall$-canon and $0$ a $\exists$-canon in the obvious symmetric manner. The following is the principle result of this section.
\begin{theorem}[Dichotomy]
Let $B$ be a boolean structure. If $B$ contains a $\forall$-canon (and a $\exists$-canon) then $\{ \exists, \forall, \wedge, \vee \}$-$\MC(B)$ is in \Logspace, otherwise it is \Pspace-complete.
\end{theorem}
\noindent The theorem follows from Propositions~\ref{prop:easy} and \ref{prop:hard} below.
\begin{examples}
Let $B_1$ and $B_2$ be the boolean structures involving the single ternary relations $\{(0,0,0),$ $(0,0,1)\}$ and $\{(0,0,0),(0,1,1)\}$, respectively. It may be verified from our classification that $\{ \exists, \forall, \wedge, \vee \}$-$\MC(B_1)$ is in \Logspace, while $\{ \exists, \forall, \wedge, \vee \}$-$\MC(B_2)$ is \Pspace-complete.
\end{examples}
\begin{proposition}
\label{prop:easy}
Let $B$ be a boolean structure. If $B$ contains a $\forall$-canon (and a $\exists$-canon) then $\{ \exists, \forall, \wedge, \vee \}$-$\MC(B)$ is in \Logspace.
\end{proposition}
\begin{proof}
In $B$, and w.l.o.g., assume that $0$ is a $\forall$-canon and $1$ is a $\exists$-canon. For $\varphi$ in  $\{ \exists, \forall, \wedge, \vee \}$-\FO, we claim that $B \models \varphi$ iff $B \models \varphi_{[\forall/0,\exists/1]}$ where $\varphi_{[\forall/0,\exists/1]}$ is the quantifier-free sentence obtained from $\varphi$ by instantiating all universal variables as $0$ and all existential variables as $1$. The evaluation of $\varphi_{[\forall/0,\exists/1]}$ on $B$ is equivalent to the Boolean Sentence Value Problem, known to be in \Logspace\ \cite{NLynch}. Our claim follows straight from the definition together with the positivity of $\varphi$. Let us consider this briefly. We may assume that all universal variables of $\varphi$, in turn, are set to $0$, since any existential witnesses to $0$ are also witnesses to $1$. Thereafter, we may assume that all remaining (existential) variables are set to $1$, because $1$ acts as a witness to everything that $0$ does.
\end{proof}

\begin{proposition}
\label{prop:hard}
Let $B$ be a boolean structure. If $B$ does not contain a $\forall$-canon, then $\{ \exists, \forall, \wedge, \vee \}$-$\MC(B)$ is \Pspace-complete.
\end{proposition}
\noindent The proof of this proposition will follow from the next three lemmas. 

\subsection{\Pspace-complete Cases}

A \emph{boolean digraph} is a boolean structure over a single binary relation $E$. Let $K_2$ and $\overline{K}_2$ be boolean digraphs with edge sets $\{(0,1),(1,0)\}$ and $\{(0,0),(1,1)\}$, respectively. The following observation will be of use to us.
\begin{lemma}
\label{lem:K2andComp}
Both $\{ \exists, \forall, \wedge, \vee \}$-$\MC(K_2)$ and $\{ \exists, \forall, \wedge, \vee \}$-$\MC(\overline{K}_2)$ are \Pspace-complete.
\end{lemma}
\begin{proof}
For $K_2$, we use a reduction from the problem $\{ \exists, \forall, \wedge, \vee \}$-$\MC(B_{NAE})$, where $B_{NAE}$ is the boolean structure with a single ternary relation $NAE:=\{0,1\}^3 \setminus \{(0,0,0),(1,1,1)\}$. This problem is a generalisation of the \emph{quantified not-all-equal $3$-satisfiability} problem -- $\{ \exists, \forall, \wedge \}$-$\MC(B_{NAE})$, a.k.a. $\QCSP(B_{NAE})$ -- well-known to be \Pspace-complete (see \cite{ComputationalComplexity}). Let $\varphi$ be an input for $\{ \exists, \forall, \wedge, \vee \}$-$\MC(B_{NAE})$. Let $\varphi'$ be built from $\varphi$ by substituting all instances of $NAE(v,v',v'')$ by $E(v,v') \vee E(v',v'') \vee E(v,v'')$. It is easy to see that $B_{NAE} \models \varphi$ iff $K_2 \models \varphi'$, and the result follows.

For $\overline{K}_2$, we reduce from the complement of $\{ \exists, \forall, \wedge, \vee \}$-$\MC(K_2)$, which we now know to be \coPspace-complete, and use the fact that $\Pspace = \coPspace$ (see \cite{ComputationalComplexity}). Let $\varphi$ be an input for $\{ \exists, \forall, \wedge, \vee \}$-$\MC(K_2)$. Generate $\varphi'$ from $\varphi$ by  swapping all instances of $\exists$ and $\forall$, and swapping all instances of $\vee$ and $\wedge$. By de Morgan's laws we may derive that $K_2 \notmodels \varphi$ iff $\overline{K}_2 \models \varphi'$, and the result follows.
\end{proof}

\begin{lemma}
Let $B$ be a boolean structure s.t. both $\neg R_B(0,\ldots,0)$ and $\neg R_B(1,\ldots,1)$. Then $\{ \exists, \forall, \wedge, \vee \}$-$\MC(B)$ is \Pspace-complete.
\end{lemma}
\begin{proof}
It follows from $(\dag)$ that $R_B$ contains some tuple $\tuple{w}:=(w_1,\ldots,w_r)$. Let $I | J$ be the partition of $\{v_1,\ldots,v_r\}$ s.t. $v_i \in I$ iff $w_i=0$. Create $R'(v_I,v_J)$ from  $R_B(v_1,\ldots,v_r)$ by identifying the variables of $I$ and $J$ as $v_I$ and $v_J$, respectively. Setting $R''(v_I,v_J):=R'(v_I,v_J) \vee R'(v_J,v_I)$ we note that $R''$ defines $K_2$. The result now follows via the obvious reduction from $\{ \exists, \forall, \wedge, \vee \}$-$\MC(K_2)$.
\end{proof}

\begin{lemma}
Let $B$ be a boolean structure s.t. both $R_B(0,\ldots,0)$ and $R_B(1,\ldots,1)$. Then $\{ \exists, \forall, \wedge, \vee \}$-$\MC(B)$ is \Pspace-complete.
\end{lemma}
\begin{proof}
It follows from $(\dag)$ that $R_B$ fails to contain some tuple $\tuple{w}:=(w_1,\ldots,w_r)$. Let $I | J$ be the partition of $\{v_1,\ldots,v_r\}$ s.t. $v_i \in I$ iff $w_i=0$. Create $R'(v_I,v_J)$ from  $R_B(v_1,\ldots,v_r)$ by identifying the variables of $I$ and $J$ as $v_I$ and $v_J$, respectively. Setting $R''(v_I,v_J):=R'(v_I,v_J) \wedge R'(v_J,v_I)$ we note that $R''$ defines $\overline{K}_2$. The result now follows via the obvious reduction from $\{ \exists, \forall, \wedge, \vee \}$-$\MC(\overline{K}_2)$.
\end{proof}

%\begin{lemma}
%Let $B$ be a boolean structure s.t. $R_B(0,\ldots,0)$ but $\neg R_B(1,\ldots,1)$, and $B$ contains no $\forall$-canon. Then $\{ \exists, \forall, \wedge, \vee \}$-$\MC(B)$ is \Pspace-complete.
%\end{lemma}
%\begin{proof}
%Knowing that $1$ is not a $\forall$-canon, we can derive the existence of some $\pi \in \Pi(r)$, some $s$ and some $w_{\pi(s)},\ldots,w_{\pi(r)} \in \{0,1\}^{r-s+1}$ s.t. $R_B(w_{\pi(1)}/1,\ldots,w_{\pi(s-1)}/1,w_{\pi(s)},\ldots,w_{\pi(r)})$ but $\neg R_B(w_{\pi(1)}/0,\ldots,w_{\pi(s-1)}/0,w_{\pi(s)},\ldots,w_{\pi(r)})$. Let $R'$ be the ternary relation obtained by identifying $w_{\pi(1)},\ldots,w_{\pi(s-1)}$ as $v_1$, and substituting the variable $v_2$ (respectively, $v_3$) in $R_B$ at all positions $i\geq s$ s.t. $w_{\pi(i)}=0$ (respectively, $w_{\pi(i)}=1$). Note that $R'(v_1,v_2,v_3)$ is s.t.
%\[
%\begin{array}{cc}
%\in R' & \notin R' \\
%(0,0,0) & (1,1,1) \\
%(1,0,1) & (0,0,1)
%\end{array}
%\]
%Consider $R''(v_1,v_3):=R_B(v_2,\ldots,v_2) \wedge R'(v_1,v_2,v_3)$. It follows that $(0,0),(1,1) \in R''$ but $(0,1) \notin R''$. Now define $R'''(v_1,v_3):=R''(v_1,v_3) \wedge R''(v_3,v_1)$. $R'''$ defines $\overline{K}_2$, and the result now follows via the obvious reduction from $\{ \exists, \forall, \wedge, \vee \}$-$\MC(\overline{K}_2)$.
%\end{proof}

\begin{lemma}
Let $B$ be a boolean structure s.t. either 
\begin{itemize}
\item[$(i)$] $\neg R_B(0,\ldots,0)$ but $R_B(1,\ldots,1)$, or
\item[$(ii)$] $R_B(0,\ldots,0)$ but $\neg R_B(1,\ldots,1)$,
\end{itemize}
and $B$ contains no $\forall$-canon. Then $\{ \exists, \forall, \wedge, \vee \}$-$\MC(B)$ is \Pspace-complete.
\end{lemma}
\begin{proof}
We prove the first case; the second follows by symmetry. 
Knowing that $0$ is not a $\forall$-canon, we can derive the existence of some partition $I | J$ of $\{v_1,\ldots,v_k\}$ s.t. $R_B(I/0,J/\tuple{w})$ but $\neg R_B(I/1,J/\tuple{w})$, where $\tuple{w}$ is some $|J|$-tuple instantiation of the elements of $J$. We further partition $J$ into $J_0|J_1$ according to whether the corresponding instantiation in $\tuple{w}$ is a $0$ or $1$. Note that each of $I$, $J_0$ and $J_1$ is non-empty. Create $R'(v_I,v_{J_0},v_{J_1})$ from  $R_B(v_1,\ldots,v_r)$ by identifying the variables in $I$, $J_0$ and $J_1$ as $v_I$, $v_{J_0}$ and $v_{J_1}$, respectively.  
Note that each of $v_I$, $v_{J_0}$ and $v_{J_1}$ appears free in $R'(v_I,v_{J_0},v_{J_1})$, which is s.t.
\[
\begin{array}{cc}
\in R' & \notin R' \\
(1,1,1) & (0,0,0) \\
(0,0,1) & (1,0,1)
\end{array}
\]
Consider $R''(v_I,v_{J_0}):=R_B(v_{J_1},\ldots,v_{J_1}) \wedge R'(v_I,v_{J_0},v_{J_1})$. It follows that $(0,0),(1,1) \in R''$ but $(1,0) \notin R''$. Now define $R'''(v_I,v_{J_0}):=R''(v_I,v_{J_0}) \wedge R''(v_{J_0},v_I)$. $R'''$ defines $\overline{K}_2$, and the result now follows via the obvious reduction from $\{ \exists, \forall, \wedge, \vee \}$-$\MC(\overline{K}_2)$.
\end{proof}

\section{Digraphs of size three}
\label{sec:digraphstart}

Let $H$ be a \emph{digraph}, that is a relational structure involving a single binary relation $E$.
For a digraph $H$, let $\overline{H}$ be the complement digraph over the same vertex set $|H|$ but with $E^{\overline{H}}:=|H|^2 \setminus E^H$. 
The following observation, essentially an extension of the second part of Lemma~\ref{lem:K2andComp}, will be of great use to us.
\begin{lemma}
\label{lem:duality}
Let $H$ be a digraph s.t. $\{ \exists, \forall, \wedge, \vee \}$-$\MC(H)$ is in \Logspace\ (respectively, is \Pspace-complete), then $\{ \exists, \forall, \wedge, \vee \}$-$\MC(\overline{H})$ is in \Logspace\ (respectively, is \Pspace-complete). Furthermore, if $\{ \exists, \forall, \wedge, \vee \}$-$\MC(H)$ is \NP-complete (respectively, is \coNP-complete), then $\{ \exists, \forall, \wedge, \vee \}$-$\MC(\overline{H})$ is \coNP-complete (respectively, is \NP-complete).
\end{lemma}
\begin{proof}
First, recall that both \Logspace\ and \Pspace\ are closed under complementation (see \cite{ComputationalComplexity}). 
Now, consider a (prenex) sentence $\psi_0$ of $\{ \exists, \forall, \wedge, \vee \}$-$\FO$. By de Morgan's laws, it is clear that $\psi_0$ is logically equivalent to the sentence $\neg \psi_1$ where $\psi_1$ is derived from $\psi$ by I.) swapping all instances of $\exists$ and $\forall$, II.) swapping all instances of $\vee$ and $\wedge$ and III.) negating all atoms (in the quantifer-free part). Let $\psi_2$ be derived from $\psi_0$ in a similar manner, but without the execution of part III (negating the atoms). It is clear that, for any digraph $H$,
\[
\begin{array}{lc}
(*) & H \models \psi_0 \ \Leftrightarrow \ H \models \neg \psi_1 \ \Leftrightarrow \ H \notmodels \psi_1 \ \Leftrightarrow \ \overline{H} \notmodels \psi_2.
\end{array}
\]
We reduce the complement of the problem $\{ \exists, \forall, \wedge, \vee \}$-$\MC(H)$ to $\{ \exists, \forall, \wedge, \vee \}$-$\MC(\overline{H})$ by the mapping $\psi_0 \mapsto \psi_2$. The results all follow from (the contrapositive of) $(*)$.
\end{proof}
In a digraph $H$, a vertex $x$ is termed a \emph{$\forall$-canon} if, for all $y \in H$, $E(x,y) \Rightarrow \forall z \ E(z,y)$ and $E(y,x) \Rightarrow \forall z \ E(y,z)$.
Dually, a vertex $x \in H$ is termed a \emph{$\exists$-canon} if, for all $y,z \in H$, $E(y,z) \Rightarrow ( E(x,z) \wedge E(y,x) )$. Note that these definitions are consistent, on boolean digraphs, with those given in Section~\ref{sec:booleanstart} (though they are given in a rather liberal notation). Being a $\forall$-canon (respectively, $\exists$-canon) is equivalent to, in the parlance of, e.g., \cite{LaroseLotenTardif}, being \emph{dominated by} (respectively, \emph{dominating})\footnote{Although this is different from the graph-theoretic notion of a dominating vertex, e.g., as used in \cite{CiE2008}.} every vertex of $H$.
It may be verified that a vertex $x \in H$ is a $\forall$-canon (respectively, $\exists$-canon) iff $x \in \overline{H}$ is a $\exists$-canon (respectively, $\forall$-canon). However, our definitions are motivated primarily by the following.
\begin{lemma}
\label{lem:canons}
Let $H$ be a digraph and let $\varphi$ be a (prenex) sentence of $\{ \exists, \forall, \wedge, \vee \}$-\FO.
\begin{itemize}
\item If $x \in H$ is a $\forall$-canon, then $H \models \varphi$ iff $H \models \varphi_{[\forall/x]}$, where $\varphi_{[\forall/x]}$ is obtained from $\varphi$ by instantiating each of the universal variables as the vertex $x$. Consequently, $\{ \exists, \forall, \wedge, \vee \}$-$\MC(H)$ is in \NP.
\item If $y \in H$ is a $\exists$-canon, then $H \models \varphi$ iff $H \models \varphi_{[\exists/y]}$, where $\varphi_{[\exists/y]}$ is obtained from $\varphi$ by instantiating each of the existential variables as the vertex $y$. Consequently, $\{ \exists, \forall, \wedge, \vee \}$-$\MC(H)$ is in \coNP.
\item If $x,y \in H$ are $\exists$-canon and $\forall$-canon, respectively, then $H \models \varphi$ iff $H \models \varphi_{[\forall/x,\exists/y]}$, where $\varphi_{[\forall/x,\exists/y]}$ is obtained from $\varphi$ by instantiating each of the universal variables as $x$ and the existential variables as $y$. Consequently, $\{ \exists, \forall, \wedge, \vee \}$-$\MC(H)$ is in \Logspace.
\end{itemize}
\end{lemma}
\begin{proof}
Recall that the Boolean Sentence Value Problem is in \Logspace\ \cite{NLynch}. All results follow straight from the definitions since $\varphi$ is positive.
\end{proof}
While the presence of both a $\forall$-canon and a $\exists$-canon is a sufficient condition for tractability of $\{ \exists, \forall, \wedge, \vee \}$-$\MC(H)$, we will see later that it is not necessary.
A vertex $x \in H$ is \emph{isolated} if, for all $y \in H$, $\neg E(x,y) \wedge \neg E(y,x)$; an isolated vertex is a $\forall$-canon.

For a digraph $H$, let $\mbox{sym-clos}(H)$ and $\mbox{tran-clos}(H)$ be the symmetric and transitive closures of $H$, respectively. Let $\mbox{doub}(H)$ be the subdigraph induced by the double edges of $H$; that is, $E^{\mbox{doub}(H)}(x,y)$ iff $E^H(x,y)$ and $E^H(y,x)$ (whereas $E^{\mbox{sym-clos}(H)}(x,y)$ iff $E^H(x,y)$ or $E^H(y,x)$). The following is a another basic observation.
\begin{lemma}
\label{lem:symtransetc}
Let $H$ be a digraph. $\{ \exists, \forall, \wedge, \vee \}$-$\MC(\mbox{sym-clos}(H))$, $\{ \exists, \forall, \wedge, \vee \}$-$\MC(\mbox{tran-clos}(H))$ and $\{ \exists, \forall, \wedge, \vee \}$-$\MC(\mbox{doub}(H))$ are all polynomial-time reducible to $\{ \exists, \forall, \wedge, \vee \}$-$\MC(H)$.
\end{lemma}
\begin{proof}
We may reduce $\{ \exists, \forall, \wedge, \vee \}$-$\MC(\mbox{sym-clos}(H))$ to $\{ \exists, \forall, \wedge, \vee \}$-$\MC(H)$ by substituting instances of $E(u,v)$ in an input $\varphi$ in the former by $E(u,v) \vee E(v,u)$ in the latter. For $\mbox{doub}(H)$ the method is similar, but the substitution is now $E(u,v)$ by $E(u,v) \wedge E(v,u)$.

For $\mbox{tran-clos}(H)$, assume $H$ is of size $n$. For $x,y \in |H|$, if there is a path in $H$ from $x$ to $y$, then there is a path from $x$ to $y$ of length $\leq n-1$. Any instances of $E(u,v)$ in an input $\varphi$ for  $\{ \exists, \forall, \wedge, \vee \}$-$\MC(\mbox{tran-clos}(H))$ should be converted to
\[
\begin{array}{ll}
\exists w^{uv}_1,\ldots,w^{uv}_{n-2} & E(u,v) \ \ \ \ \ \ \vee \\
 & E(u,w^{uv}_1) \wedge E(w^{uv}_1,v) \ \ \ \ \ \ \vee \\
 & \vdots \\
 & E(u,w^{uv}_1) \wedge E(w^{uv}_1,w^{uv}_2) \wedge \ldots \wedge E(w^{uv}_{n-3},w^{uv}_{n-2}) \wedge E(w^{uv}_{n-2},v) \\
\end{array}
\]
in an instance of $\{ \exists, \forall, \wedge, \vee \}$-$\MC(\mbox{tran-clos}(H))$.
\end{proof}
For a digraph $H$ and $x \in H$, define $H \setminus \{x\}$ to be the induced subdigraph of $H$ on vertex set $|H| \setminus \{x\}$. We will also need the following result.
\begin{lemma}
\label{lem:equfree}
Let $x,y \in H$ be vertices that satisfy, for all $z \in H$, $E(x,z) \Leftrightarrow E(y,z)$ and $E(z,x) \Leftrightarrow E(x,y)$. Then $\{\neg, \exists, \forall, \wedge, \vee \}$-$\MC(H)=$ $\{\neg, \exists, \forall, \wedge, \vee \}$-$\MC(H \setminus \{x\})$.
\end{lemma}
\begin{proof}
Intuitively, FO logic without equality can not distinguish between $x$ and $y$, and so can not tell if only one of them is there. More formally, one observes that the surjective homomorphism $h:H \rightarrow H \setminus \{x\}$ given by $x \mapsto y$ together with the identity on $H \setminus \{x\}$ has the property that it preserves negated (as well as positive) atoms. For full details see, e.g, the Homomorphism Theorem in \cite{Enderton}.
\end{proof}
Let $K_1=K^0_1$, $K_2=K^{00}_2$, $K_3=K^{000}_3$, $K^1_1$, $K^{11}_2$ and $K^{111}_3$ be the complete antireflexive digraphs on $1$, $2$ and $3$ vertices and complete reflexive digraphs on $1$, $2$ and $3$ vertices, respectively.  Let $P^{00}_2$, $P^{000}_3$, $DP^{00}_2$ and $DP^{000}_3$ denote the antireflexive undirected $1$- and $2$-paths and the antireflexive directed $1$- and $2$-paths, respectively. The superscripted $1$s and $0$s indicate vertices with or without self-loops, respectively, whence the meaning of, say, $DP^{100}_3$ as a directed $2$-path whose first vertex is the only self-loop, should become clear. We will also build non-connected digraphs from the disjoint union of certain of these. Note that our digraphs may have multiple notations under our various conventions, e.g. $K_2=P^{00}_2$, $\overline{K}_2=K^1_1 \uplus K^1_1$ and $P^{01}_2=P^{10}_2$ (although $DP^{01}_2 \neq DP^{10}_2$). 
\begin{proposition}
\label{prop:basicresults}
The following basic results will form the backbone of our tetrachotomy.
\[
\begin{array}{lll}
(i) & \{ \exists, \forall, \wedge, \vee \}\mbox{-}\MC(K_2) & \mbox{is \Pspace-complete} \\
(ii) & \{ \exists, \forall, \wedge, \vee \}\mbox{-}\MC(\overline{K}_2) & \mbox{is \Pspace-complete} \\
(iii) & \{ \exists, \forall, \wedge, \vee \}\mbox{-}\MC(K_3) & \mbox{is \Pspace-complete} \\
(iv) & \{ \exists, \forall, \wedge, \vee \}\mbox{-}\MC(\overline{K}_3) & \mbox{is \Pspace-complete} \\
(v) & \{ \exists, \forall, \wedge, \vee \}\mbox{-}\MC(K_1 \uplus K_2) & \mbox{is \NP-complete} \\
(vi) & \{ \exists, \forall, \wedge, \vee \}\mbox{-}\MC(P^{000}_3) & \mbox{is \Pspace-complete} \\
(vii) & \{ \exists, \forall, \wedge, \vee \}\mbox{-}\MC(K^1_1 \uplus K^{11}_2) & \mbox{is \Pspace-complete}
\end{array}
\]
\end{proposition}
\begin{proof}
$(i)$ and $(ii)$. Are proved in Lemma~\ref{lem:K2andComp}.

$(iii)$. $\{ \exists, \forall, \wedge, \vee \}$-$\MC(K_3)$ contains the problem $\{ \exists, \forall, \wedge \}$-$\MC(K_3)$, a.k.a. $\QCSP(K_3)$, as a special instance. The latter is well-known to be \Pspace-complete (see \cite{OxfordQuantifiedConstraints}).

$(iv)$. Follows from $(iii)$ via Lemma~\ref{lem:duality}. 

$(v)$. For $\{ \exists, \forall, \wedge, \vee \}\mbox{-}\MC(K_1 \uplus K_2)$, note that the vertex of $K_1$ is a $\forall$-canon and the problem $\{ \exists, \forall, \wedge, \vee \}$-$\MC(K_1 \uplus K_2)$ is in \NP\ by Lemma~\ref{lem:canons}. For completeness, note that the problems $\{ \exists, \wedge, \vee \}$-$\MC(K_1 \uplus K_2)$ and $\{ \exists, \wedge, \vee \}$-$\MC(K_2)$ coincide (that is, $K_1 \uplus K_2$ and $K_2$ agree on all sentences of $\{ \exists, \wedge, \vee \}$-\FO\ -- see \cite{CiE2008}). The \NP-complete problem \emph{not-all-equal $3$-satisfiability} may be reduced to $\{ \exists, \wedge, \vee \}$-$\MC(K_2)=\{ \exists, \wedge, \vee \}$-$\MC(K_1 \uplus K_2)$, as in the second part of Lemma~\ref{lem:K2andComp}, so \NP-hardness of the superproblem $\{ \exists, \forall, \wedge, \vee \}$-$\MC(K_1 \uplus K_2)$ of $\{ \exists, \wedge, \vee \}$-$\MC(K_1 \uplus K_2)$ immediately follows.

$(vi)$ and $(vii)$. \Pspace-completeness of $\{ \exists, \forall, \wedge, \vee \}$-$\MC(K^1_1 \uplus K^{11}_2)$ and $\{ \exists, \forall, \wedge, \vee \}$-$\MC(P^{000}_3)$ now follows from Lemma~\ref{lem:equfree} since $\overline{K}_2=K^{1}_1 \uplus K^{1}_1$ and $K^1_1 \uplus K^{11}_2$ (respectively, $P^{000}_3$ and $K_2$) agree on all sentences of $\{\neg, \exists, \forall, \wedge, \vee \}$-\FO.
\end{proof}
The two digraphs $H$ of size $1$ clearly give rise to $\{ \exists, \forall, \wedge, \vee \}$-$\MC(H)$ in \Logspace. The classification for digraphs $H$ of size $2$ may be read from that for boolean structures (it is also explicitly in \cite{CiE2008}) as a dichotomy between those $\{ \exists, \forall, \wedge, \vee \}$-$\MC(H)$ that are in \Logspace, and those that are \Pspace-complete. We are now in a position to work through the main result of this section.
\begin{theorem}[Tetrachotomy]
\label{thm:tetrachotomy}
Let $H$ be a digraph of size $3$. Then $\{ \exists, \forall, \wedge, \vee \}$-$\MC(H)$ is either in \Logspace, is \NP-complete, is \coNP-complete or is \Pspace-complete.
\end{theorem}
\noindent We will prove this theorem through exhaustive consideration of a variety of cases. We may refer back to the known cases of Proposition~\ref{prop:basicresults} without citation.

\subsection{Digraphs that are either non-connected, antireflexive or reflexive}

\paragraph{Non-connected digraphs.}
Let $H$ be a non-connected digraph. We consider two cases. 

$H$ contains an isolated vertex $x$. Since $x$ is a $\forall$-canon, we know $\{ \exists, \forall, \wedge, \vee \}$-$\MC(H)$ is in \NP\ (by Lemma~\ref{lem:canons}). It is complete if the other component is non-empty and antireflexive, since then $\symclos(H)$ is $K_1 \uplus K_2$ (see Lemma~\ref{lem:symtransetc}). If the other component is empty, then $H:=K_1 \uplus K_1 \uplus K_1$ and $\{ \exists, \forall, \wedge, \vee \}$-$\MC(H)$ is in \Logspace\ (there are no yes-instances). Suppose now that the other component contains a self-loop at $y$. Since $x$ is an isolated vertex and a $\forall$-canon, we know that $H \models \varphi$ iff $H \models \varphi_{[\forall/x]}$ (where $\varphi_{[\forall/x]}$ is $\varphi$ with the universal variables evaluated to $x$), but now it is clear from the self-loop at $y$ that $H \models \varphi_{[\forall/x]}$ iff $H \models \varphi_{[\forall/x\exists/y]}$ (where $\varphi_{[\forall/x,\exists/y]}$ is $\varphi_{[\forall/x]}$ with the remaining (existential) variables evaluated to $y$). It follows that $\{ \exists, \forall, \wedge, \vee \}$-$\MC(H)$ is also in \Logspace\ in this case.

$H$ contains no isolated vertex. In this case $\tranclos(\symclos(H))$ is either $\overline{K}_3$ or $K^{1}_1 \uplus K^{11}_2$ and $\{ \exists, \forall, \wedge, \vee \}$-$\MC(H)$ is \Pspace-complete, by Lemma~\ref{lem:symtransetc}. 

\paragraph{Connected antireflexive digraphs.}
If $H$ is antireflexive and connected, then $\symclos(H)$ is either $P^{000}_3$ or $K_3$, and $\{ \exists, \forall, \wedge, \vee \}$-$\MC(H)$ is \Pspace-complete by Lemma~\ref{lem:symtransetc}.

\paragraph{Reflexive digraphs.}
Reflexive digraphs' complements are antireflexive, and may be classified, through Lemma~\ref{lem:duality}, according to the previous two paragraphs.

\subsection{Connected digraphs with one or two self-loops that are subdigraphs of $P^{111}_3$}
\label{sec:2.1}

\paragraph{One self-loop at end.} 
All digraphs $H$ in this category are s.t. $\symclos(H)$ is the digraph $P^{100}_3$, where $\{ \exists, \forall, \wedge, \vee \}$-$\MC(\overline{P^{100}_3})$ is \Pspace-complete (since $\tranclos(\overline{P^{100}_3})$ is $K^{11}_2 \uplus K^{1}_1$). It follows from Lemmas~\ref{lem:duality} and \ref{lem:symtransetc} that $\{ \exists, \forall, \wedge, \vee \}$-$\MC(H)$ is \Pspace-complete.

\vspace{0.2cm} \hspace{0.5cm} \input{p100_3_2.pstex_t}

\paragraph{One self-loop in the middle.}
For $H:=$ $P^{010}_3$, $H_1$ or $H'_1$, drawn below, the problem $\{ \exists, \forall, \wedge, \vee \}$-$\MC(H)$ is in \Logspace. This is due to these $H$ agreeing on all sentences of $\{ \neg, \exists, \forall, \wedge, \vee \}$-$\FO$ with those respective digraphs on two vertices drawn to their right (see Lemma~\ref{lem:equfree}). The result now follows from Lemma~\ref{lem:canons} as the vertices $a$ and $b$ on the right-hand digraphs are $\forall$-canon and $\exists$-canon, respectively.

\vspace{0cm} \hspace{1cm} \input{self_loop_in_middle_1.pstex_t}

\noindent When $H$ is either of the following $H_2$ or $H'_2$, the vertices $c$ and $b$ are $\forall$-canon and $\exists$-canon, respectively. It follows from Lemma~\ref{lem:canons} that the problem $\{ \exists, \forall, \wedge, \vee \}$-$\MC(H)$ is in \Logspace\ in both cases.

\hspace{1cm} \input{self_loop_in_middle_2.pstex_t}

\noindent We have only one more digraph to consider in this paragraph: $DP^{010}_3$, drawn below with its complement.

\vspace{.2cm} \hspace{2cm} \input{self_loop_in_middle_3.pstex_t}

\noindent In $DP^{010}_3$, the vertex $b$ is a $\exists$-canon; it follows from Lemma~\ref{lem:canons} that $\{ \exists, \forall, \wedge, \vee \}$-$\MC(DP^{010}_3)$ is in  \coNP. We now show that it is complete by demonstrating the \NP-hardness of $\{ \exists, \forall, \wedge, \vee \}$-$\MC(\overline{DP^{010}_3})$. Consider the following relation on $\overline{DP^{010}_3}$,
\[
\begin{array}{l}
\forall w \ E(w,w) \vee ( E(u,w) \wedge E(w,v) ) \ \ \ \vee \\
\forall w \ E(w,w) \vee ( E(w,u) \wedge E(w,v) ),
\end{array}
\]
which defines $K_1 \uplus K_2$ (note that the ``$\forall w \ E(w,w) \vee \ldots$'' may be more easily read as ``$\forall w \ \neg E(w,w) \rightarrow \ldots$''). Since $\{ \exists, \forall, \wedge, \vee \}$-$\MC(K_1 \uplus K_2)$ is \NP-complete, the result follows from Lemma~\ref{lem:duality}.

\paragraph{Two self-loops, none in the middle.} 
All digraphs $H$ in this category are s.t. $\symclos(H)$ is the digraph $P^{101}_3$, where $\tranclos(\overline{P^{101}_3})$ is $K^{11}_2 \uplus K^{1}_1$. It follows from Lemmas~\ref{lem:duality} and \ref{lem:symtransetc} that $\{ \exists, \forall, \wedge, \vee \}$-$\MC(H)$ is \Pspace-complete.

\vspace{0cm} \hspace{.5cm} \input{p101_3_2.pstex_t}

\paragraph{Two self-loops, one in the middle.} 
Firstly, we consider the digraph $P^{110}_3$ and four of its subdigraphs. 

\vspace{.2cm} \hspace{1cm} \input{p110_3_2.pstex_t}

\noindent For $H:=$ $P^{110}_3$, $H_3$ or $H'_3$, $c$ is a $\forall$-canon and $b$ is a $\exists$-canon. It follows from Lemma~\ref{lem:canons} that $\{ \exists, \forall, \wedge, \vee \}$-$\MC(H)$ is in \Logspace. 

For $H:=$ $H_4$ or $H'_4$, we can only say that $c$ is a $\forall$-canon ($b$ is not actually a $\exists$-canon). For $\varphi$ in $\{ \exists, \forall, \wedge, \vee \}$-\FO, let $\varphi_{[\forall/c]}$ be $\varphi$ with the universal variables evaluated to $c$ and let $\varphi_{[\forall/c,\exists/b]}$ be $\varphi_{[\forall/c]}$ with the remaining (existential) variables evaluated to $b$. We know from Lemma~\ref{lem:canons} that $H_4 \ (H'_4) \models \varphi$ iff $H_4 \ (H'_4) \models \varphi_{[\forall/c]}$. It is easy to see that $H_4 \ (H'_4) \models \varphi_{[\forall/c]}$ iff $H_4 \ (H'_4) \models \varphi_{[\forall/c,\exists/b]}$, by the positivity of $\varphi$, since the third vertex $a$ has no adjacency to $c$.
It follows that $\{ \exists, \forall, \wedge, \vee \}$-$\MC(H_4)$ and $\{ \exists, \forall, \wedge, \vee \}$-$\MC(H'_4)$ are both in \Logspace.
We now turn our attention to the following twins.

\vspace{.2cm} \hspace{1cm} \input{twinh5s.pstex_t}

\noindent For $H:=$ $H_5$ or $H'_5$, we have that $\doub(\tranclos(H))$ is $K^{11}_2 \uplus K^{1}_1$; and we may deduce in both cases, from Lemma~\ref{lem:symtransetc}, that $\{ \exists, \forall, \wedge, \vee \}$-$\MC(H)$ is \Pspace-complete.

\noindent This leaves us with the twins $DP^{110}_3$ and $DP^{011}_3$, over which we may define $H_5$ and $H'_5$ by
\[ 
\begin{array}{l}
E(u,v) \ \vee \ \forall w \ E(w,w) \vee ( E(v,w) \wedge \exists w' \ E(w',w) \wedge E(w',u) ) \mbox{ and} \\
E(v,u) \ \vee \ \forall w \ E(w,w) \vee ( E(w,u) \wedge \exists w' \ E(w,w') \wedge E(v,w') ),
\end{array}
\]
respectively.
It follows that both $\{ \exists, \forall, \wedge, \vee \}$-$\MC(DP^{110}_3)$ and $\{ \exists, \forall, \wedge, \vee \}$-$\MC(DP^{011}_3)$ are \Pspace-complete.

\vspace{0.2cm} \hspace{1cm} \input{dp110_3.pstex_t}

\subsection{Connected digraphs with one or two self-loops that are not subdigraphs of $P^{111}_3$}

\paragraph{Digraphs with a double edge.}
The complements of these are either non-connected or are connected subdigraphs of $P^{111}_3$, and may be classified accordingly.

\paragraph{Orientations of $K^{100}_3$ and $K^{110}_3$.} It remains only to consider the following eight digraphs, $H_6$, $H_7$, $H'_7$, $H_8$ drawn below with their respective complements.

\hspace{.5cm} \input{orientationsofk3.pstex_t}

\noindent We define the following relation on $\overline{H_6}$.
\[ E(u,v) \ \vee \ \forall w \ E(w,w) \vee ( \exists w' E(w,w') \wedge E(w',u) \wedge \exists w'' E(w'',w) \wedge E(w'',v) ) \]
This relation defines the digraph $\overline{DP^{100}_3}$, where we know $\{ \exists, \forall, \wedge, \vee \}$-$\MC(DP^{100}_3)$ is \Pspace-complete (see first paragraph of Section~\ref{sec:2.1}). It follows, by Lemma~\ref{lem:duality}, that $\{ \exists, \forall, \wedge, \vee \}$-$\MC(\overline{H_6})$ is \Pspace-complete, and also that $\{ \exists, \forall, \wedge, \vee \}$-$\MC(H_6)$ is \Pspace-complete.

\noindent For $\overline{H_8}$, the vertex $b$ is a $\forall$-canon, so by Lemma~\ref{lem:canons} $\{ \exists, \forall, \wedge, \vee \}$-$\MC(\overline{H_8})$ is in \NP. For completeness, define the digraph $K_1 \uplus K_2$ over $\overline{H_8}$ as follows.
\[
\begin{array}{l}
\forall w \ E(w,w) \vee ( E(u,w) \wedge E(w,v) ) \ \ \ \vee \\
\forall w \ E(w,w) \vee ( E(w,u) \wedge E(v,w) )
\end{array}
\]
It follows from Lemma~\ref{lem:duality} that $\{ \exists, \forall, \wedge, \vee \}$-$\MC(H_8)$ is \coNP-complete.

\noindent Finally, we turn our attention to the cases $H_7$, $\overline{H_7}$, $H'_7$ and $\overline{H'_7}$. By proving that $\{ \exists, \forall, \wedge, \vee \}$-$\MC(\overline{H_7})$ is in \Logspace, we may deduce, via Lemma~\ref{lem:duality}, the like result for $H_7$ (the proof is the same for $\overline{H'_7}$ and $H'_7$). In $\overline{H_7}$, none of the vertices is either a $\forall$-canon or a $\exists$-canon. But, we will deduce the following, which directly implies that $\{ \exists, \forall, \wedge, \vee \}$-$\MC(\overline{H_7})$ is in \Logspace.
\begin{itemize}
\item[$(\ddagger)$] If $\varphi$ is a sentence of $\{ \exists, \forall, \wedge, \vee \}$-\FO, then $\overline{H_7} \models \varphi$ iff $\overline{H_7} \models \varphi_{[\forall/b,\exists/c]}$, where $\varphi_{[\forall/b,\exists/c]}$ is obtained from $\varphi$ by instantiating all of the universal variables as $b$ and all the existential variables as $c$.
\end{itemize}
(Proof of $(\ddagger)$.) Let $\varphi_{[\forall/b]}$ (respectively, $\varphi_{[\exists/c]}$) be obtained from $\varphi$ by instantiating all universal variables as $b$ (respectively, all existential variables as $c$).

(Forwards.) 
Assume $\overline{H_7} \models \varphi$. If all universal variables are set to $b$, then it follows that (existential) witnesses for $\overline{H_7} \models \varphi_{[\forall/b]}$ exist. Note that, in $\overline{H_7}$, for all $v$ we have both $E(b,v) \Rightarrow E(b,c)$ and $E(v,b) \Rightarrow E(c,b)$ (the latter is vacuously true). It follows by the positivity of $\varphi$ that we may assume those witnesses are all $c$.

(Backwards.) 
We claim that $\overline{H_7} \models \varphi_{[\forall/b,\exists/c]}$ implies $\overline{H_7} \models \varphi_{[\exists/c]}$ (which, a fortiori, gives $\overline{H_7} \models \varphi$). The claim is true due to the positivity of $\varphi$ since, in $\overline{H_7}$, for all $v$ we have both $E(b,c) \Rightarrow E(v,c)$ and $E(c,b) \Rightarrow E(c,v)$ (again the latter is vacuously true).

\section{Further Work}

Our ultimate goal is to extend the tetrachotomy of Theorem~\ref{thm:tetrachotomy} to all digraphs. This would give the like tetrachotomy for arbitrary finite relational structures (see \cite{FederVardi}). %It may be easier to derive the result for undirected graphs. Note that it follows from our classification that this class restricted to size three also displays tetrachotomy.

\bibliographystyle{acm}
\bibliography{Combined2}

\begin{thebibliography}{10}

\bibitem{OxfordQuantifiedConstraints}
{\sc B{\"o}rner, F., Krokhin, A., Bulatov, A., and Jeavons, P.}
\newblock Quantified constraints and surjective polymorphisms.
\newblock Tech. Rep. PRG-RR-02-11, Oxford University, 2002.

\bibitem{BulatovJACM}
{\sc Bulatov, A.~A.}
\newblock A dichotomy theorem for constraint satisfaction problems on a
  3-element set.
\newblock {\em J. ACM 53}, 1 (2006), 66--120.

\bibitem{Nadia}
{\sc Creignou, N., Khanna, S., and Sudan, M.}
\newblock {\em Complexity classifications of Boolean Constraint Satisfaction
  Problems}.
\newblock SIAM Monographs. 2001.

\bibitem{Enderton}
{\sc Enderton, H.~B.}
\newblock {\em {A Mathematical Introduction to Logic}}.
\newblock Academic Press, 1972.

\bibitem{FederVardi}
{\sc Feder, T., and Vardi, M.~Y.}
\newblock The computational structure of monotone monadic {SNP} and constraint
  satisfaction: a study through datalog and group theory.
\newblock {\em SIAM J. Comput. 28\/} (1999).

\bibitem{HellNesetril}
{\sc Hell, P., and Ne\v{s}et\v{r}il, J.}
\newblock On the complexity of {H}-coloring.
\newblock {\em J. Combin. Theory Ser. B 48\/} (1990).

\bibitem{jeavons98algebraic}
{\sc Jeavons, P.}
\newblock On the algebraic structure of combinatorial problems.
\newblock {\em Theoretical Computer Science 200}, 1--2 (1998), 185--204.

\bibitem{Ladner}
{\sc Ladner, R.~E.}
\newblock On the structure of polynomial time reducibility.
\newblock {\em J. ACM 22}, 1 (1975), 155--171.

\bibitem{LaroseLotenTardif}
{\sc Larose, B., Loten, C., and Tardif, C.}
\newblock A characterisation of first-order constraint satisfaction problems.
\newblock In {\em LICS 2006\/} (2006), IEEE Computer Society, pp.~201--210.

\bibitem{NLynch}
{\sc Lynch, N.}
\newblock Log space recognition and translation of parenthesis languages.
\newblock {\em Journal of the ACM 24\/} (1977), 583--590.

\bibitem{CiE2008}
{\sc Martin, B.}
\newblock First order model checking problems parameterized by the model.
\newblock In {\em CiE 2008, LNCS 5028\/} (2008), pp.~417--427.

\bibitem{DBLP:conf/cie/MartinM06}
{\sc Martin, B., and Madelaine, F.~R.}
\newblock Towards a trichotomy for quantified {H}-coloring.
\newblock In {\em CiE 2006, LNCS 3988\/} (2006), pp.~342--352.

\bibitem{ComputationalComplexity}
{\sc Papadimitriou, C.}
\newblock {\em Computational Complexity}.
\newblock Addison-Wesley, 1994.

\bibitem{VardiComplexity}
{\sc Vardi, M.}
\newblock Complexity of relational query languages.
\newblock In {\em 14th Symposium on Theory of Computation\/} (1982).

\end{thebibliography}

\end{document}